%% file: main.tex
\documentclass[a4paper,cleveref, autoref, thm-restate]{lipics-v2021}
\usepackage{tikz}
\usepackage{mathtools}
\usepackage{amsthm}
\usepackage{booktabs}
\usepackage{array}
\usepackage{xspace}
\usepackage{ifthen}
\usepackage{graphicx}
\usepackage{todonotes}
\usepackage{nicefrac}
\usepackage{bm}
\usepackage{cite} %
\usepackage{colortbl}

\usetikzlibrary{intersections}
\usetikzlibrary{math,calc}
\usetikzlibrary{decorations.shapes,decorations.pathmorphing}

\title{Crossing Numbers of Beyond Planar Graphs Re-revisited: A Framework Approach} %

\author{Markus Chimani}{Theoretical Computer Science, Osnabrück University, Germany}{markus.chimani@uos.de}{https://orcid.org/0000-0002-4681-5550}{}%
\author{Torben Donzelmann}{Theoretical Computer Science, Osnabrück University, Germany}{tdonzelmann@uos.de}{}{}
\author{Nick Kloster}{Theoretical Computer Science, Osnabrück University, Germany}{nkloster@uos.de}{}{}
\author{Melissa Koch}{Theoretical Computer Science, Osnabrück University, Germany}{melikoch@uos.de}{}{}
\author{Jan-Jakob Völlering}{Theoretical Computer Science, Osnabrück University, Germany}{jvoellering@uos.de}{}{}
\author{Mirko H.\ Wagner}{Theoretical Computer Science, Osnabrück University, Germany}{mirko.wagner@uos.de}{https://orcid.org/0000-0003-4593-8740}{}

\authorrunning{M. Chimani, T. Donzelmann, N. Kloster, M. Koch, J. Völlering, and M.\,H. Wagner} %

\Copyright{Markus Chimani, Torben Donzelmann, Nick Kloster, Melissa Koch, Jan-Jakob Völlering, and Mirko H.\ Wagner} %

\ccsdesc[500]{Mathematics of computing~Graph theory}

\keywords{Beyond planarity, crossing number, crossing ratio, proof framework} %

\category{} %

\relatedversion{} %

\EventEditors{...}
\EventNoEds{2}
\EventLongTitle{...}
\EventShortTitle{...}
\EventAcronym{...}
\EventYear{...}
\EventDate{...}
\EventLocation{...}
\EventLogo{}
\SeriesVolume{..}
\ArticleNo{0}

\tikzstyle{blueE}=[blue,dash pattern=on 12pt off .6pt]
\tikzstyle{yellowE}=[yellow!50!brown,dash pattern=on 12pt off .6pt on 1.2pt off .6pt on 1.2pt off .6pt]
\tikzstyle{redE}=[red,dash pattern=on 9pt off .6pt on 1.2pt off .6pt]
\tikzstyle{grayE}=[gray]%

\input{graphs.tex}

\newcommand{\ksharp}[1][1]{#1-jumping} %
\newcommand{\ksharpness}{1-jumping} %
\newcommand{\redSymbol}{'}

\newcommand{\C}{\ensuremath{\mathfrak{c}}\xspace}
\newcommand{\Calt}{\ensuremath{\tilde{\mathfrak{c}}}\xspace}
\newcommand{\D}[1][]{\ensuremath{\mathcal{D}#1}\xspace}
\newcommand{\Dred}{\D[\redSymbol]}
\renewcommand{\O}{\ensuremath{\mathcal{O}}\xspace}
\newcommand{\Gbpnn}[1][\bp]{\ensuremath{\mathcal{G}_{#1}}\xspace}
\newcommand{\Gbp}[1][\bp]{\ensuremath{\mathcal{G}_{#1}(n)}\xspace}
\newcommand{\Gbpd}[1][\bp]{\ensuremath{\mathcal{G}^{\mathrm{dense}}_{#1}(n)}\xspace}
\newcommand{\Gbps}[1][\bp]{\ensuremath{\mathcal{G}^{\mathrm{sparse}}_{#1}(n)}\xspace}

\newcommand{\connections}{\ensuremath{C}\xspace}

\newcommand{\kecrParam}{k}
\newcommand{\conceptsym}{\mathfrak{B}}
\newcommand{\bp}[1][(k)]{\ensuremath{\conceptsym#1}\xspace}
\newcommand{\kpl}[1][k]{\ensuremath{#1\mathrm{\text{-}pl}}\xspace}

\newcommand{\kgappl}[1][k]{\ensuremath{#1\mathrm{\text{-}gap\text{-}pl}}\xspace}
\newcommand{\kfcf}[1][k]{\ensuremath{#1\mathrm{\text{-}fcf}}\xspace}
\newcommand{\slrac}{\ensuremath{\mathrm{sl\text{-}RAC}}\xspace}

\newcommand{\IC}{\ensuremath{\mathrm{\text{}IC}}\xspace}
\newcommand{\NIC}{\ensuremath{\mathrm{\text{}NIC}}\xspace}
\newcommand{\kvp}[1][k]{\ensuremath{#1\mathrm{\text{-}vp}}\xspace}
\newcommand{\NNIC}{\ensuremath{\mathrm{\text{}NNIC}}\xspace}
\newcommand{\skewk}[1][k]{\ensuremath{\mathrm{skew\text{-}}#1}\xspace}
\newcommand{\kapex}[1][k]{\ensuremath{#1\mathrm{\text{-}apex}}\xspace}
\newcommand{\ac}{\ensuremath{\mathrm{\text{}ac}}\xspace}
\newcommand{\fc}{\ensuremath{\mathrm{\text{}fc}}\xspace}
\newcommand{\wfp}{\ensuremath{\mathrm{\text{}wfp}}\xspace}
\newcommand{\sfp}{\ensuremath{\mathrm{\text{}sfp}}\xspace}

\newcommand{\kecr}[1][k]{\ensuremath{#1\mathrm{\text{-}ec}}\xspace}
\newcommand{\crossingnum}[1]{\ensuremath{\mathrm{cr}_{#1}}}
\newcommand{\crn}{\ensuremath{\crossingnum{}}\xspace}
\newcommand{\bcr}{\overline{\mathrm{cr}}}
\newcommand{\bcrbp}[1][\bp]{\ensuremath{\bcr_{#1}}}
\newcommand{\crbp}[1][\bp]{\ensuremath{\crossingnum{#1}}\xspace}
\newcommand{\crkpl}[1][k]{\ensuremath{\crossingnum{\kpl[#1]}}\xspace}

\newcommand{\crkvp}[1][k]{\ensuremath{\crossingnum{\kvp[#1]}}\xspace}
\newcommand{\crsfp}[1][k]{\ensuremath{\crossingnum{\sfp}}\xspace}
\newcommand{\crwfp}[1][k]{\ensuremath{\crossingnum{\wfp}}\xspace}
\newcommand{\crfc}[1][k]{\ensuremath{\crossingnum{\fc}}\xspace}
\newcommand{\crac}[1][k]{\ensuremath{\crossingnum{\ac}}\xspace}

\newcommand{\crkecr}[1][k]{\ensuremath{\crossingnum{\kecr[#1]}}\xspace}
\newcommand{\crkfcf}[1][k]{\ensuremath{\crossingnum{\kfcf[#1]}}\xspace}

\newcommand{\crIC}{\ensuremath{\crossingnum{\IC}}\xspace}
\newcommand{\crNIC}{\ensuremath{\crossingnum{\NIC}}\xspace}

\newcommand{\bcrr}[2][(n)]{\ensuremath{\overline{\varrho}_{#2}#1}\xspace}
\newcommand{\bcrrbp}[1][(n)]{\ensuremath{\bcrr[#1]{\bp}}\xspace}
\newcommand{\crr}[2][(n)]{\ensuremath{\varrho_{#2}#1}\xspace}

\newcommand{\crrbp}[1][(n)]{\ensuremath{\crr[#1]{\bp}}\xspace}
\newcommand{\width}[1]{\ensuremath{w(#1)}\xspace}
\newcommand{\widthRed}[1]{\ensuremath{w'(#1)}\xspace}

\newcommand{\widthRkfcfD}{\ensuremath{2k}\xspace}

\newcommand{\Gl}[1][\ell]{\ensuremath{G_{#1}}\xspace}
\newcommand{\Kuras}[1][]{\ensuremath{\mathcal{K}_{#1}}\xspace}
\newcommand{\KurasRed}{\ensuremath{\mathcal{K}_{\Dred}}\xspace}
\newcommand{\KS}[1][s]{Kuratowski subdivision#1\xspace}
\newcommand{\otks}[1][]{of the \Kuras[#1]-subdivisions\xspace}
\newcommand{\otksRed}{of the \KurasRed-subdivisions\xspace}
\newcommand{\col}{\mathit{col}}
\newcommand{\cols}{\mathfrak{col}}
\newcommand{\Rs}[1][\bp]{\crr[^{\mathrm{sparse}}(n)]{#1}}

\newcommand{\bundle}[3][]{\ensuremath{\ifthenelse{\equal{#3}{}}{(#2,2)^{#1}}{(#2,#3)^{#1}}}\text{-bundle}\xspace}

\begin{document}
\nolinenumbers
\makeatletter
\long\def\@secondoffive#1#2#3#4#5{#2}
\long\def\@fifthoffive#1#2#3#4#5{#5}
\makeatother
\maketitle

\begin{abstract}
Beyond planarity concepts (prominent examples include $k$-planarity or fan-planarity) apply certain restrictions on the allowed patterns of crossings in drawings. It is natural to ask, how much the number of crossings may increase over the traditional (unrestricted) crossing number. Previous approaches to bound such ratios, e.g.~\cite{CKMV19,BPS21}, require very specialized constructions and arguments for each considered beyond planarity concept, and mostly only yield asymptotically non-tight bounds.
We propose a very general proof framework that allows us
to obtain asymptotically tight bounds, and
where the concept-specific parts of the proof typically boil down to a couple of lines.
We show the strength of our approach by giving improved or first bounds for several beyond planarity concepts.
\end{abstract}

\section{Introduction}

Throughout this paper, we only consider \emph{simple} graphs, i.e., no self-loops or multi-edges. Given a graph $G$, let $n$ and $m$ denote its number of vertices $V(G)$ and edges $E(G)$, respectively. A drawing $\mathcal{D}$ of $G$ is a mapping of $V(G)$ to distinct points in $\mathbb{R}^2$, and $E(G)$ to curves connecting the respective end points. Such a curve must not contain any vertex point other than its end points. When two edge curves intersect on an internal point $x$, this is a \emph{crossing} of these two edges at $x$. The \emph{crossing number} $\crn(G)$ is the smallest number of crossings over all drawings. To achieve it, we can safely assume \emph{simple drawings}, i.e., no three edges cross at a common point, no edge crosses itself or adjacent edges, and no pair of edges crosses multiple times. A graph is \emph{planar} if it allows a crossing-free drawing. In the last decades, several \emph{beyond planarity} concepts have been established that generalize planar graphs in that certain special crossing patters are allowed or forbidden, see, e.g.,~\cite{SurveyBeyondplanarity,beyondBook} for overviews. A prominent example is $k$-planarity, where any edge may be crossed at most $k$ times.

We are interested in the minimum number of crossings within any beyond planarity restricted drawing. Intuitively speaking, a beyond planarity concept is a way to (try to) formalize our understanding of what aspects constitute a readable drawing. Still, within these restrictions it is natural to ask for a drawing with the least number of crossings. Alternatively, we can ask how much it ``costs'' in terms of crossings to follow a certain drawing paradigm.

Formally, let \bp denote any beyond planarity concept, where we omit the parameter~$k$ for parameterless beyond planarity concepts. Then, $\Gbpnn$ ($\Gbp$) denotes the set of graphs (on $n$ vertices, respectively) allowing a \bp-drawing.
The \emph{\bp-crossing number} $\crbp(G)$ for a graph $G\in\Gbpnn$ is the least number of crossings over all \bp-drawings of $G$. The \emph{crossing ratio} $\crrbp$ of \bp is the largest attainable ratio between the \bp-crossing number and the (normal) crossing number, over all $n$-vertex \bp-graphs:
\[\crrbp \coloneqq \sup_{G \in \Gbp} \frac{\crbp(G)}{\crn(G)}.\]

Crossing ratios of beyond planarity concepts have for the first time been explicitly considered in~\cite{CKMV19a,CKMV19}, where linear lower bounds $\Omega(n)$ where established for 1-planar, $k$-quasi-planar, and (weakly) fan-planar graphs. In all but the first case, the upper bounds where (at least) quadratic in $n$, and it was conjectured in~\cite{CKMV19} that the real lower bounds should rather be $\Omega(n^2)$. Later,~\cite{BPS21} generalized the results to $k$-planarity and seven more beyond planarity concepts. However, none of the provided bounds were tight except for planarly-connected and straight-line RAC. In both publications, each planarity measure needs a very specific construction and intricate specialized arguments to prove the lower bounds (which tend to always be the most complicated part in crossing number proofs).
Furthermore, both only consider simple drawings, when establishing upper bounds on the crossing ratio. Simple and non-simple crossing ratios may differ, if there are \bp-graphs that only allow non-simple \bp-drawings. Our bounds below work with and without the simplicity assumption.

\subparagraph*{Our contribution.}
In this paper, we propose a general framework that, when applicable, simplifies crossing ratio proofs down to only a couple of \bp-specific lines (10--20 lines for the lower bound).
In contrast to previous schemes, this framework is able to prove asymptotically tight bounds for all considered beyond planarity concepts, and our proofs further show that the crossing ratio is already achieved between any two subsequent parameterizations $k$ and $k+1$, in all considered parametrized concepts.
\Cref{tab:res} summarizes our results. A key idea lies in the simplification of the necessary lower bound arguments, by turning them into counting arguments over a set of (sufficiently) disjoint Kuratowski-subdivisions (see \Cref{sec:prelim}). In particular, this allows us to avoid any intricate discussions about alternative drawing possibilities, the simplicity of drawings, etc.
We present the proof framework in \Cref{sec:framework}, and showcase its versatility and strength for various beyond planarity concepts in \Cref{sec:use}.

\begin{table}
\newcommand{\Ryes}{\checkmark}
\newcommand{\Rno}{$\times$}
\newcommand{\markA}{$^*$}
\newcommand{\markB}{}
\caption{Overview on our results on the crossing ratios of beyond planarity concepts for $n$-vertex graphs.
Formally, \cite{CKMV19} considers weakly fan-planar, but their proof works for all 4 variants.
The column ``$\overline{\varrho}$?'' denotes whether we also show the (same) rectilinear crossing ratios in \Cref{sec:rectilinear}.
The $^*$-marked previous upper bounds need no explicit discussion of non-simple drawings by the nature of the beyond planarity concept.
All our bounds, except the upper bounds for adjacency-crossing and fan-crossing, work with and without the simplicity assumption.}
\label{tab:res}
\centering
\setlength{\aboverulesep}{0pt}
\setlength{\belowrulesep}{0pt}
\setlength{\extrarowheight}{1.5pt}
\begin{tabular}{l@{}c@{}c>{\columncolor{black!5}}crc}
\toprule
    beyond planarity concept        & previous best                             &                & our results                & Section          & $\overline{\varrho}$? \\
\hline
$k$-planar                  & $\Omega(n/k)\cap\O(k\sqrt{k}n)$ \markA    & \cite{BPS21}   & $\Theta(n)$                & \ref{sec:kpl}    & \Ryes                 \\
$k$-vertex-planar           & ---                                       &                & $\Theta(n)$                & \ref{sec:kvp}    & \Ryes                 \\
IC-planar                   & ---                                       &                & $\Theta(n)$                & \ref{sec:ic}     & \Ryes                 \\
NIC-planar                  & ---                                       &                & $\Theta(n)$                & \ref{sec:nic}    & \Ryes                 \\
NNIC-planar                 & ---                                       &                & $\Theta(n^2)$              & \ref{sec:nnic}   & \Ryes                 \\
$k$-fan-crossing-free       & $\Omega(n^2/k^3)\cap\O(k^2n^2)$ \markA    & \cite{BPS21}   & $\Theta(n^2/k)$            & \ref{sec:nnic}   & \Ryes                 \\
straight-line RAC           & $\Theta(n^2)$ \markA                      & \cite{BPS21}   & \emph{(direct corollary)}  & \ref{sec:nnic}   & ---                   \\
$\begin{array}{@{}l} 
\text{adjacency-crossing}  \\
\text{fan-crossing}        \\
\text{weakly fan-planar}   \\
\text{strongly fan-planar} \\ \end{array}\hspace{-1em}
\left.\begin{array}{@{}l}  \\                                          \\               \\                           \\ \end{array}\right\rbrace$
                            & $\Omega(n)\cap\O(n^2)$ \markB             & \cite{CKMV19}  & $\Theta(n^2)$              & \ref{sec:fp}     & \Rno                  \\
$k$-edge-crossing           & ---                                       &                & $\Theta(k)$                & \ref{sec:ecr}    & \Ryes                 \\
$k$-gap-planar              & $\Omega(n/k^3)\cap\O(k\sqrt{k}n)$ \markA  & \cite{BPS21}   & $\Theta(n/k)$              & \ref{sec:kgappl} & \Ryes                 \\
$k$-apex-planar             & $\Omega(n/k)\cap\O(k^2n^2)$ \markB    & \cite{BPS21}   & $\Theta(n^2/k)$            & \ref{sec:kapex}  & \Ryes                 \\
skewness-$k$                & $\Omega(n/k)\cap\O(kn)$ \markB        & \cite{BPS21}   & $\Theta(n)$                & \ref{sec:kapex}  & \Ryes                 \\
\bottomrule
\end{tabular}
\end{table}

\section{Preliminaries}\label{sec:prelim}
\subparagraph*{Kuratowski subdivisions.} A graph $G$ is planar if and only if it does not contain a $K_5$ or $K_{3,3}$ subdivision (summarily called \emph{Kuratowski subdivisions})~\cite{kuratowski1930probleme}. Let $K$ be a $K_{3,3}$ subdivision in $G$. The six degree-3 vertices of $K$ are the \emph{Kuratowski nodes} of $K$, and the paths between them (not containing other Kuratowski nodes) its \emph{Kuratowski paths}; by replacing each Kuratowski path by a single edge, one obtains $K_{3,3}$. Two Kuratowski paths are \emph{adjacent}, if they share a common Kuratowski node.
Clearly, we have $\crn(G)\geq 1$, and any drawing \D of $G$ (even a non-simple one) contains at least one crossing $x$ between two edges from \emph{non-adjacent} Kuratowski paths. We say that $x$ \emph{covers}~$K$.
Every Kuratowski subdivision in $G$ has to be covered in \D; a single crossing may cover several such subdivisions.

\subparagraph*{Upper bounds.}
The \emph{crossing lemma}~\cite{CrossingLemmaPaper} states that any graph $G$ with $n$ vertices and $m>4n$ edges has $\crn(G)\in\Omega(\frac{m^3}{n^2})$.
We may split $\Gbp$ into graphs $\Gbpd\coloneqq\{G\in\Gbp \mid m > 4n\}$ that are sufficiently dense for the crossing lemma, and $\Gbps\coloneqq\Gbp\setminus\Gbpd$.
It will later turn out that we attain the largest ratios on graphs of the latter subset. Let $\Rs\coloneqq\sup_{G \in \Gbps}\frac{\crbp(G)}{\crn(G)}$.
    Then 
    \[
\crrbp 
= \max\left\{ \Rs \ ,\  \sup_{G \in \Gbpd} \frac{\crbp(G)}{\crn(G)} \right\},\]
where we want to use the crossing lemma to bound the second term in the maximum.
\begin{observation}\label{thm:ub}
    We have:
    \begin{enumerate}
        \item[(a)] If $\crbp(G)\in \O(m^2)$, then $\crrbp \in {\color{gray}\O(\Rs + m^2\cdot\frac{n^2}{m^3}) \subseteq }\; \O(\Rs+n).$
        \item[(b)] If $\crbp(G)\in \O(mk)$, then
  $\crrbp \in {\color{gray}\O(\Rs + mk\cdot\frac{n^2}{m^3}) \subseteq }\; \O(\Rs+k).$
    \end{enumerate}
\end{observation}

All simple drawings satisfy the prerequisite of \cref{thm:ub}(a).
Furthermore, in the context of connected sparse graphs in $\Gbps$ (and thus for $\Rs$), we know $m\in\Theta(n)$.

\subparagraph*{\ksharpness.}
The crossing ratio $\crrbp$ relates the \bp-restricted crossing number to the normal one. It is natural to ask, how small a parameter $\Delta\geq1$ can be chosen such that $\crr[^{\Delta}(n)]{\bp}\coloneqq
\sup_{G\in \Gbp} \frac{\crbp(G)}{\crbp[{\bp[(k+\Delta)]}](G)}$ attains the same asymptotic bound.
We say $\crrbp$ is \emph{\ksharp[$\Delta$]} if $\crr[^{\Delta}(n)]{\bp}\in\Omega(\crrbp)$.
For all concepts in \Cref{sec:use} except $k$-gap-planarity (\cref{sec:kgappl}), we will observe \emph{\ksharpness} as our upper bounds on the normal crossing numbers are already attained by $\bp[(k+1)]$-drawings.

\section{Framework for Proving Lower Bounds on Crossing Ratios}\label{sec:framework}
\begin{figure}
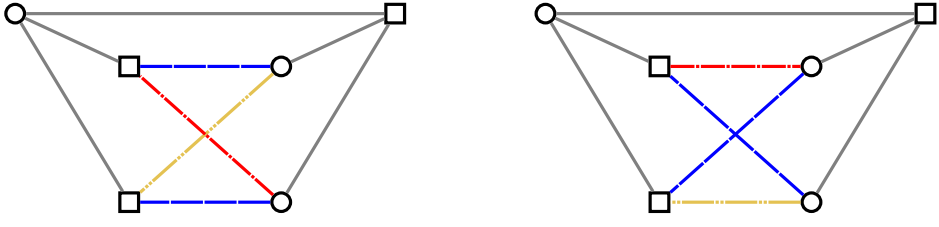

	\centering
	\begin{subfigure}[b]{0.49\textwidth}
		\centering
        \frameDrawing{blueE}{blueE}{yellowE}{redE}{}{}
	\end{subfigure}
	\begin{subfigure}[b]{0.49\textwidth}
		\centering
        \frameDrawing{yellowE}{redE}{blueE}{blueE}{swapped}{}
	\end{subfigure}
	\caption{Red-yellow \textbf{\textsf{(left)}} and blue-blue \textbf{\textsf{(right)}} drawing of $F$ used for standard drawings of~$\Gl$.}
	\label{fig:base}
\end{figure}

Let $\bp$ be a beyond planarity concept, and $\psi\colon \mathbb{N}^2 \to \mathbb{N}$ a function dependent on $k$ and the graph's number of vertices $n$.
To show that the crossing ratio $\varrho_{\bp}\geq \psi(k,n)$, we construct an infinite family $\{\Gl\}_{\ell\geq\ell_0}$, for some $\ell_0\in\mathbb{N}$, such that for each \Gl we have
\[\frac{\crbp[\bp](\Gl)}{\crn(\Gl)}\geq\psi(k,|V(\Gl)|).\]

To construct this family, we always start with a \emph{frame} $F=(N,C,\cols)$, an edge-colored $K_{3,3}$.
To avoid ambiguity with the graphs constructed based on $F$, we call elements of $N\coloneqq\{v_1,v_2,v_3,w_1,w_2,w_3\}$ \emph{nodes} and elements of $C\coloneqq \{ \{v_i,w_j\} \mid i,j\in\{1,2,3\}\}$ \emph{connections}.
Each connection $\mathfrak{c}\in C$ has a color $\cols(\mathfrak{c})$ that is either \emph{red}, \emph{blue}, \emph{yellow}, or \emph{gray}.

Based on this frame $F$ and some parameter $\ell$, we construct \emph{framework graphs} \Gl.
To this end, we define for each of the four colors $\col$ a graph $S_{\col}$ (which we will call \emph{con-graph} in the following) with two pole vertices $s,t$.
The size of $S_{\col}$ often depends on $\ell$, $k$, or both.
Intuitively, we may talk about, e.g., a \emph{red} con-graph if it is associated with the color red.
We then replace each connection $\{a,b\}\in C$ with a new copy $S_{\{a,b\}}$ of its color's con-graph $S_{\cols(\{a,b\})}$, identifying this con-graphs' poles $s$ and $t$ with $a$ and $b$, respectively.

In most of our proofs, the con-graphs are rather simple: an \emph{$(i,j)$-bundle} is a con-graph consisting of $i$ internally-vertex-disjoint $s$-$t$-paths, each of length $j$ (i.e., $j$ edges in each path); an \emph{$(i,j)^+$-bundle} furthermore contains the edge $\{s,t\}$.
In all our constructions, the yellow con-graph is the single edge $\{s,t\}$, which can equivalently be seen as a $(1,1)$-bundle.

\subparagraph*{Standard Drawings of \boldmath$\Gl$.}
In our framework, we need to show (i) that $\Gl$ is indeed in \Gbp, (ii) an upper bound on $\crn(\Gl)$, and (iii) a lower bound on $\crbp(\Gl)$.
We show (i) and (ii) by constructing \emph{standard drawings} of $\Gl$:
These drawings are constructed by first drawing the frame $F$; in particular this maps connections to curves in the plane.
We obtain a drawing of $\Gl$ by drawing each con-graph in a small neighborhood of its corresponding connection's curve (instead of drawing the connection itself).
Whenever possible, we draw con-graphs planarly with their poles on their local outer face; otherwise we explicitly specify their drawing.
Apart from this, we only need to specify how to draw crossing con-graphs.

Unless specified otherwise, our frame $F$ contains a 4-cycle of a blue, a yellow, another blue, and a red connection in this order; its remaining five connections are gray.
Thus, to cover the $K_{3,3}$, there has to be a red-yellow crossing, a blue-blue crossing, or a crossing involving a gray connection.
We can classify drawings of $F$ with a single crossing by the colors of the crossing connections; we are in particular interested in a red-yellow and a blue-blue drawing of $F$, cf.\ \cref{fig:base}.
Such drawings yield corresponding standard drawings of $\Gl$.
We use a red-yellow standard drawing to show an upper bound on $\crn(\Gl)$. As the yellow con-graph is only a single edge, the number of crossings is kept low.
However, the red con-graph is chosen such that this drawing is not a \bp-drawing.
We use a blue-blue standard drawing to show that $\Gl$ is in \Gbp. However, the blue con-graphs are chosen such that many crossings arise.
Lastly, the gray con-graphs are chosen such that a standard drawing with a crossed gray connection is not a \bp-drawing and also requires many crossings.

\subparagraph*{Lower bounding \boldmath$\crbp(\Gl)$.}
To prove a lower bound on $\crbp(\Gl)$, we aim to show that the number of crossings in the drawing establishing $\Gl\in\Gbp$ is asymptotically optimal.
Until now, we only considered standard drawings of $\Gl$, where each con-graph was treated as a unit.
For the proof of the lower bound on $\crbp(\Gl)$, we consider an \emph{arbitrary} \bp-drawing \D of \Gl. In particular, therein con-graphs may, for example, partially cross themselves or one another.
For every connection $\C \in \connections$, let $P_{\C}$ be a (not necessarily maximal) set of $s$-$t$-paths in the con-graph $S_\C$.
If a con-graph $S_\C$ is a bundle, then we use the set of all its edge-disjoint paths as $P_\C$; otherwise we will define it specifically. 
Let $\width{\C} \coloneqq \width{S_\C} \coloneqq |P_\C|$ be the \emph{width} of $S_\C$; further, let the \emph{height} $h(\C) \coloneqq h(S_\C) \coloneqq \max_{p \in P_\C} |p|$ be the number of edges in the longest path in $P_\C$.
For an edge $e \in S_\C$, let $P_{\C}[e]$ be those paths in $P_{\C}$ that contain $e$.
We define $\Kuras\coloneqq\Kuras({G_\ell,N)}$ as the set of all $K_{3,3}$ subdivisions of $\Gl$ that have $N$ as their Kuratowski nodes and where each Kuratowski path is from $\bigcup_{\C \in \connections} P_{\C}$.
There is a mapping between the Kuratowski paths of $\Kuras$ and the connections $\connections$ of $F$, such that each Kuratowski path consists of an $s$-$t$-path in the respective connection's con-graph.

Consider some \KS[] $K \in \Kuras$ that is covered in \D by a crossing $x$ between $e_1,e_2 \in E(G)$.
For $i\in\{1,2\}$, let $\C_i \in \connections$ be the connection with $e_i \in S_{\C_i}$. The connections $\C_1$ and $\C_2$ are non-adjacent in $F$.
The crossing $x$ may cover all \KS that have a Kuratowski path from $P_{\C_1}[e_1]$ and one from $P_{\C_2}[e_2]$, but no other of~\Kuras.
By definition of \Kuras, the crossing $x$ covers at most $\frac{|P_{\C_1}[e_1]| \cdot |P_{\C_2}[e_2]|}{|P_{\C_1}| \cdot |P_{\C_2}|}$ of the \KS in \Kuras.
In most cases (in particular whenever $S_{\C_1},S_{\C_2}$ are both bundles) we have $|P_{\C_1}[e_1]|=|P_{\C_2}[e_2]|=1$, and then the crossing $x$ covers $\frac{1}{|P_{\C_1}| \cdot |P_{\C_2}|}$ \otks.

    In our proofs, we sometimes consider a specific subdrawing of \D.
    Let $\C\in C$ be a connection with con-graph $S_\C$ and let $R$ be an $s$-$t$-path in $S_\C$.
    Subdrawing \Dred is the \emph{$R$-restriction} of \D where all of 
    $S_\C$ except for $R$ (a bundle of width~$1$) is removed. Within \Dred, $R$ is considered to be the con-graph of \C, and we are interested in the minimum number of crossings covering the Kuratowski subdivisions $\KurasRed\subseteq\Kuras$ in \Dred. We may consider an $(R_1,R_2)$-restriction, where we perform this operation twice, for distinct connections. 

\section{Proving Crossing Ratios}\label{sec:use}
We now use our framework to prove the crossing ratio for several beyond planarity concepts. 
Generally, whenever we define such a concept \bp, we do so via \bp-drawings. 
This is \emph{always} to be understood to implicitly define the corresponding graph class as those graphs that allow such a \bp-drawing: e.g., a graph is $k$-planar if it allows a $k$-planar drawing. See \Cref{tab:res} for previously established crossing ratios (if any).

\subsection{\texorpdfstring{$\boldmath\text{}k$}{k}-Planar\texorpdfstring{: $\boldmath\crr{\kpl}$}{}}
\label{sec:kpl}
In a $k$-planar ($k\geq 1$, shorthand ``$k$-pl'') drawing, no edge is crossed more than $k$ times~\cite{kplanarFirstIntroducedPaper,PT97}.

\begin{figure}
	\centering
	\input{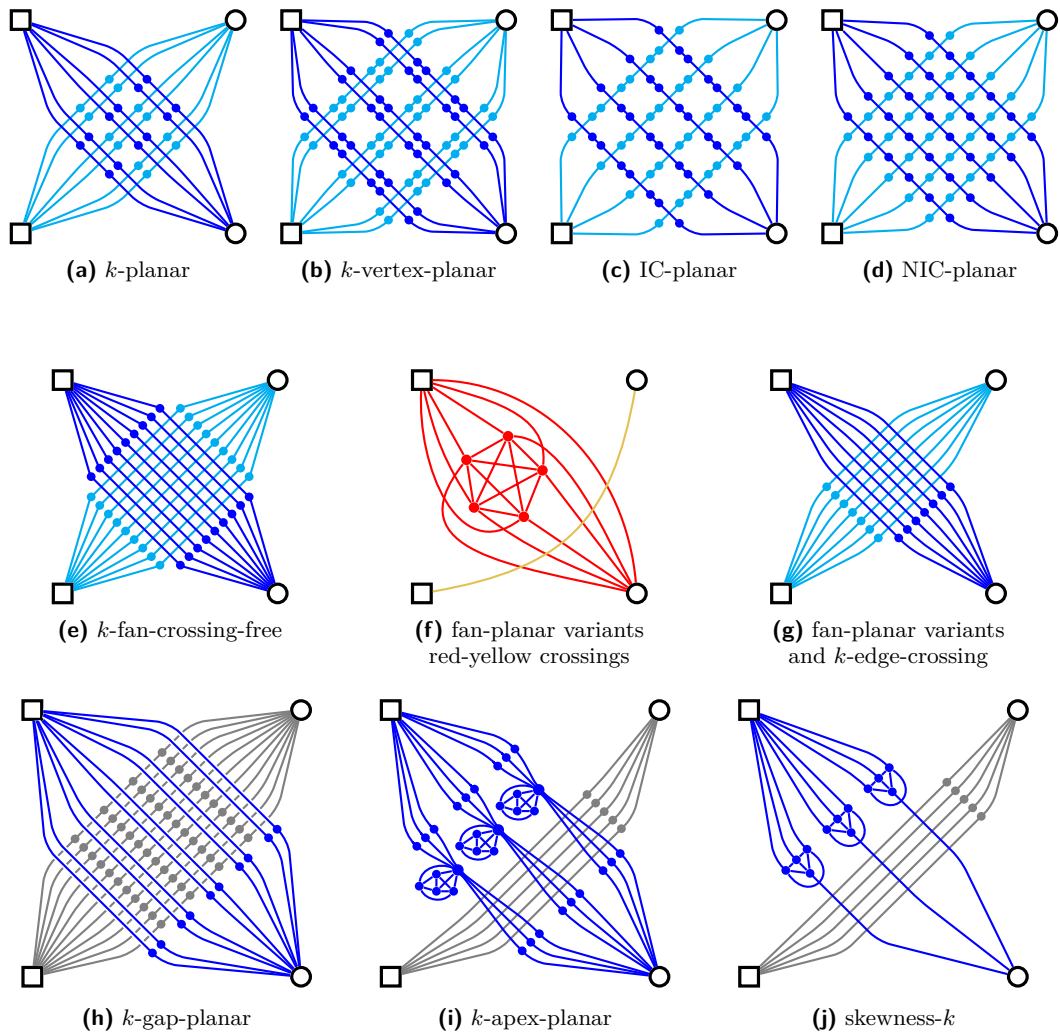}
 	\caption{Visualizations of standard drawings of crossing con-graphs, for different beyond planarity concepts. The con-graphs are colored consistently to their color type (blue, gray, red, yellow); we use two shades of blue when showing crossing blue con-graphs. 
 The depicted values for $\ell$ are chosen for easy comprehension of the drawing paradigm (e.g., \textbf{\textsf{(a)}} shows $\ell=3$, $k=2$); they are smaller than required in the proofs.}
	\label{fig:crossings}
\end{figure}

\begin{theorem}\label{thm:kpl:lb}
    For every $\ell \geq 41$, there exists a $k$-planar graph \Gl with $n\in \Theta(\ell^2 k)$ vertices such that $\crkpl(\Gl) \in \Omega(nk)$ and $\crn(\Gl) \in \O(k)$. Thus $\crr{\kpl} \in \Omega(n)$.
\end{theorem}
\begin{proof}
    Let \Gl be a framework graph, where the red con-graph is a $\bundle{k+1}{}$, the blue con-graphs are $\bundle{\ell k}{\ell}$s, and the gray con-graphs are $\bundle{\ell k}{2}$s, so $n \in \Theta(\ell^2 k)$.
    A red-yellow standard drawing shows $\crn(\Gl) \in \O(k)$.
    A blue-blue standard drawing, see \cref{fig:kpl_blue_blue_cr} for how the blue con-graphs cross, shows that \Gl is $k$-planar.

    Let \D be a $k$-planar drawing of \Gl. At least one $s$-$t$-path $R$ in the red con-graph is not crossed by the unique yellow edge, as the latter has at most $k$ crossings.
    From now on, we thus only consider the $R$-restriction \Dred of \D, which does not have any red-yellow crossings.
    Let $\C, \Calt$ be two connections.
    Any crossing between edges from $\C$ and $\Calt$ can cover at most $\frac{1}{\width{\C} \cdot \width{\Calt}}$ \otksRed (recall that the width of the red con-graph here in \Dred is $1$).
    By $k$-planarity, there can be at most $k \cdot \min\{|E(\C)|, |E(\Calt)|\}$ crossings between edges from $\C$ and $\Calt$.
    As all con-graphs are bundles, the number of edges in each con-graph is the product of its height and width.
    Thus, such crossings can overall cover at most $k \cdot \frac{\min\{\width{\C} \cdot h(\C), \width{\Calt} \cdot h(\Calt)\}}{\width{\C} \cdot \width{\Calt}}$ \otksRed.
    Let $\C$ be a gray connection, then crossings between $\C$ and another connection $\Calt$ can cover at most
    \[k \cdot \frac{\min\{\ell k \cdot 2, \width{\Calt} \cdot h(\Calt)\}}{\ell k \cdot \width{\Calt}} =
    \left.
    \begin{cases}
        \frac{1    \cdot 1}{\ell \cdot 1} & \text{if $\Calt$ is yellow}\\
        \frac{1   \cdot 2}{\ell \cdot 1} & \text{if $\Calt$ is red}\\
        \frac{\ell k \cdot 2}{\ell \cdot \ell k} & \text{if $\Calt$ is blue}\\
        \frac{\ell k \cdot 2}{\ell \cdot \ell k} & \text{if $\Calt$ is gray}\\
    \end{cases}\right\}\leq \frac{2}{\ell}\]
    \otksRed.
    Five con-graphs are gray and four connections are non-adjacent to a given connection, so crossings involving a gray edge can cover at most $5 \cdot 4 \cdot \frac{2}{\ell} \leq \frac{40}{41}$ \otksRed.
    Thus, blue-blue crossings cover at least $\frac{1}{41}$ \otksRed.
    Since each such crossing covers at most $\frac{1}{(\ell k)^2}$ \otksRed, this yields a total of $\Omega((\ell k)^2)$ crossings in \D and so $\crkpl(G) \in \Omega(n k)$ and $\crr{\kpl} \in \Omega(n)$.
\end{proof}

A non-$k$-planar drawing has an edge that is crossed at least $k+1$ times, and $\crkpl(\Gl)\leq mk$, as in a $k$-planar drawing each edge is crossed at most $k$ times.
Thus, \Cref{thm:kpl:lb} and \Cref{thm:ub}(b) with $\Rs[\kpl]\leq \frac{mk}{k+1}\in \O(n)$ yield:
\begin{corollary}
	\label{thm:kpl:crr}
	The $k$-planar crossing ratio $\crr{\kpl}$ is in $\Theta(n)$ and \ksharp.
\end{corollary}

\subsection{\texorpdfstring{$\boldmath\text{}k$}{k}-Vertex-Planar\texorpdfstring{: $\boldmath\crr{\kvp}$}{}}
\label{sec:kvp}
A vertex $v$ is \emph{adjacent} to a crossing $x$, if $x$ is on an edge incident to $v$. 
In a \emph{$k$-vertex-planar} (shorthand ``$\kvp$'') drawing, no vertex may be adjacent to more than $k$ crossings.
\begin{theorem}\label{thm:kvp:lb}
    For every $\ell \geq 11$, there exists a $k$-vertex-planar graph \Gl with $n\in \Theta(\ell^2 k)$ vertices such that $\crkvp(\Gl) \in \Omega(n k)$ and $\crn(\Gl) \in \O(k)$. Thus $\crr{\kvp} \in \Omega(n)$.
\end{theorem}
\begin{proof}
    Let \Gl be a framework graph, where the red con-graph is a $\bundle{k+1}{}$s, the blue con-graphs are $\bundle{\ell k}{2\ell+1}$s, and the gray con-graphs are $\bundle{\ell k}{}$s, so $n \in \Theta(\ell^2 k)$.
    A red-yellow standard drawing shows $\crn(\Gl) \in \O(k)$.
    A blue-blue standard drawing shows that \Gl is $k$-vertex-planar, cf.\ \cref{fig:kvp_blue_blue_cr}.

    Let \D be a $k$-vertex-planar drawing of \Gl. Since the yellow edge $e$ has at most $k$ crossings, there exists an $s$-$t$-path $R$ in the red con-graph not crossed by $e$.
    Let \Dred be the $R$-restriction of \D; it does not have any red-yellow crossings.
    Let $\C$ be a gray connection with con-graph $S_\C$; it has width $\ell k$. Since any crossing on $S_\C$ is adjacent to one of the poles of $S_\C$, there are at most $2k$ crossings on $S_\C$ and thus
    crossings involving at least one edge from $S_\C$ can cover at most $\frac{2k}{\ell k}=\frac{2}{\ell}$ \otksRed.
    Five con-graphs are gray, so crossings involving a gray edge can cover at most $5 \cdot \frac{2}{\ell} \leq \frac{10}{11}$ \otksRed.
    Thus, blue-blue crossings cover at least $\frac{1}{11}$ \otksRed.
    Since each such crossing covers at most $\frac{1}{(\ell k)^2}$ \otksRed, this yields $\crkvp(G) \in \Omega((\ell k)^2) = \Omega(nk)$.
\end{proof}

A non-$k$-vertex-planar drawing has at least $k+1$ crossings, and $\crkvp(G)\leq nk$, as every vertex may be adjacent to at most $k$ crossings.
Thus, \Cref{thm:kvp:lb} and \Cref{thm:ub}(b) with $\Rs[\kvp]\leq \frac{nk}{k+1}\in \O(n)$ yield:
\begin{corollary}
	\label{thm:kvp:crr}
	The $k$-vertex-planar crossing ratio $\crr{\kvp}$ is in $\Theta(n)$ and \ksharp.
\end{corollary}

\subsection{IC-Planar\texorpdfstring{: $\boldmath\crr{\IC}$}{}}
\label{sec:ic}
Two crossings \emph{share a vertex} $v$ if they are both on edges incident to $v$. A drawing is \emph{IC-planar} (\emph{independent crossings} planar, shorthand ``IC''), if no two crossings share a vertex~\cite{ic-intro1,ic-intro2,nic-intro2}.
Note that \crIC must not be mistaken for the \emph{independent crossing number}, where there are no restrictions on the crossings, but one only counts crossings between non-adjacent edges~\cite{MarcusSchaeferCrNoSurvey}.
In fact, IC-planarity is equivalent to $1$-vertex-planarity.
We can obtain near-tight non-asymptotic bounds for this special case.

\begin{theorem}
    \label{thm:ic:lb}
    For every $\ell \geq 2$, there exists an IC-planar graph \Gl with $n=4\ell^2+12$ vertices such that $\crIC(\Gl) = \frac{n}{4}-3$ and $\crn(G)\leq 2$. Thus $\crr{\IC}\geq \frac{n}{8}-\frac32$.
\end{theorem}
\begin{proof}
    Let $\Gl$ be a framework graph, where the red and gray con-graphs are $\bundle[+]{1}{2}$s (i.e., triangles), and the blue con-graphs are $\bundle{\ell}{2\ell+1}$s, so $n = 4\ell^2+12$.
    The graph has $n =6+6+2\cdot2\ell\cdot \ell = 4\ell^2+12$ vertices.
    A red-yellow standard drawing shows that $\crn(G_\ell) \leq 2$. A blue-blue standard drawing shows that the graph is IC-planar, cf.\ \Cref{fig:ic_blue_blue_cr}. 

    Consider an IC-planar drawing $\D$ of $\Gl$.
    For each red or gray connection $\C=\{s,t\}$,
    at least one of the two $s$-$t$-paths in its con-graph has to be uncrossed, as otherwise there would be non-independent crossings.
    Since the yellow connection is adjacent to the blue connections, all \KS have to be covered by blue-blue crossings. Since blue con-graphs have width $\ell$, any blue-blue crossing can only cover at most $\frac{1}{\ell^2}$ \otks, yielding $\crn(G_\ell) \geq \ell^2 = \frac{n}{4}-3$.
\end{proof}

A non-IC-planar drawing has at least $2$ crossings, and $\crIC(G)\leq n/4$, as no crossings share a vertex, thus $\crr{\IC}\leq \frac{n/4}2=\frac{n}8$. Together with \Cref{thm:ic:lb} we have:

\begin{corollary}
    \label{thm:ic:crr}
    For the IC-planar crossing ratio it holds that $\crr{\IC}\in [\frac{n}{8}-\frac32, \frac{n}{8}]$.
\end{corollary}

We note that by shrinking the length of a single path in each of the two blue con-graphs by $2$, our construction still works and (since then $n=4\ell^2+8$) we would get the slightly tighter lower bound of $\frac{n}{8}-1$, at the expense of a slightly longer proof.

 \subsection{NIC-Planar\texorpdfstring{: $\boldmath\crr{\NIC}$}{}}
\label{sec:nic}
In an \emph{NIC-planar} (\emph{nearly independent crossings}, shorthand ``NIC'') drawing, any two crossings may share at most one vertex~\cite{nic-intro1, nic-intro2}. 
Clearly, each edge is involved in at most one crossing.

\begin{theorem}
\label{thm:nic:lb}
For every $\ell\geq 4$, there exists an NIC-planar graph $\Gl$ with 
$n \in \Theta(\ell^2)$ vertices such that $\crNIC(\Gl)\in \Omega(n)$ and $\crn(\Gl) \leq 2$. Thus $\varrho_{IC}(n) \in \Omega(n)$.
\end{theorem}
\begin{proof}
    Let \Gl be a framework graph where the red con-graph is a \bundle[+]{1}{}, the blue con-graphs are $\bundle{\ell}{\ell+2}$s, and the gray con-graphs are $\bundle{\ell}{2}$s, so $n\in \Theta(\ell^2)$.
    A red-yellow standard drawing shows $\crn(\Gl) \leq 2$.
    A blue-blue standard drawing, shows that \Gl is NIC-planar, cf.\ \Cref{fig:nic_blue_blue_cr}.
    
    Consider a NIC-planar drawing \D of \Gl. Since the yellow edge can be crossed at most once and all other con-graphs have width at least 2, a crossing on the yellow edge covers at most $\frac{1}{2}$ \otks.
    Any red-gray crossing can cover at most $\frac{1}{2\ell}$ \otks. As each of the three red edges is involved in at most one crossing, such crossings cover at most $\frac{3}{2\ell}$ \otks in total. Therefore, at least $1- \frac{1}{2}- \frac{3}{2\ell}\geq \frac{1}{8}$ \otks have to be covered by crossings involving only gray and blue edges.
    Any such crossing covers at most $\frac{1}{\ell^2}$ \otks, yielding $\crNIC(\Gl)\in \Omega(\ell^2)=\Omega(n)$. 
\end{proof}

A non-NIC-planar drawing has at least $2$ crossings, and $\crIC(G)\leq m/2$, as each crossing requires two edges not involved in any other crossing.
Since NIC-planar graphs have at most $\frac{18}{5}n$ edges~\cite{NIC-Planar},
we have $\crr{\NIC} \leq \frac{9}{10}n$ and together with \Cref{thm:nic:lb}:
\begin{corollary}
    \label{thm:nic:crr}
    The NIC-planar crossing ratio $\crr{\NIC}$ is in $\Theta(n)$.
\end{corollary}

\subsection{NNIC-Planar and \texorpdfstring{$\boldmath\text{}k$}{k}-Fan-Crossing-Free\texorpdfstring{: $\boldmath\crr{\NNIC}$ and $\boldmath\crr{\kfcf}$}{}}
\label{sec:nnic}
A simple drawing is \emph{NNIC-planar} (\emph{nearly nearly independent crossings}, shorthand ``NNIC''), if any two crossings share at most two vertices.
This beyond planarity concept has seemingly not been investigated before, but it is closely related to \emph{fan-crossing-free} (FCF) drawings~\cite{kFanCrossingFreePaper}: simple drawings where we disallow a single edge to cross over a fan, i.e., multiple edges that share a common incident vertex. Observe that if an edge $\{x,y\}$ would cross two edges incident to a common vertex $z$, these two crossings would share the three vertices $x,y,z$.
Thus, NNIC may at first sight seem identical to FCF. However, while
every NNIC-planar graph is clearly also FCF, the inverse is not necessarily the case:

\begin{restatable}{theorem}{THMnnicfcf}
\label{theoremNNICNotEqualToFanCrossing}
    There exist fan-crossing-free graphs that are not NNIC-planar. 
\end{restatable}
\begin{proof}[Proof sketch]
  \Cref{fig:nnic_vs_fcf} shows an FCF graph that is not NNIC-planar.
 We refer to the
 appendix
 for a detailed proof of that fact.
  A central idea is that copies of $K_5$ can be used as ``blocking walls'' as they cannot be crossed in an FCF drawing. 
\end{proof}

\begin{figure}
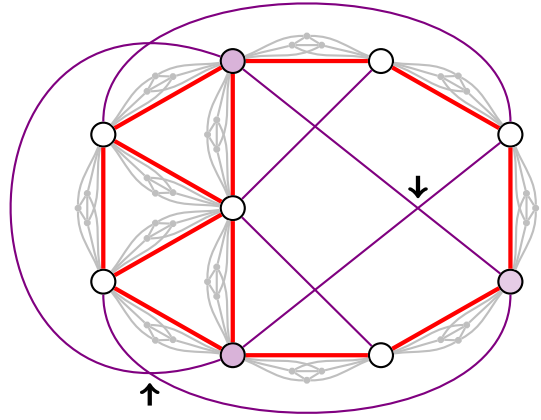

\centering
        \FCFvsNNICdrawing{false}
    \caption{A fan-crossing-free but not NNIC-planar graph. The red edges cannot be crossed in any FCF drawing. The depicted fan-crossing-free subembedding of the non-gray edges is unique (up to mirroring), but not NNIC-planar: arrows point to two crossings that share three (shaded) vertices.}
    \label{fig:nnic_vs_fcf}
\end{figure}

More generally, we can consider the generalization of FCF graphs, the \emph{$k$-fan-crossing-free} (shorthand ``$k$-fcf'') graphs, $k\geq 2$, where any edge may cross at most $k-1$ edges of a fan in a simple drawing~\cite{kFanCrossingFreePaper}. Standard FCF is thus identical to $2$-fcf.

\begin{theorem}\label{thm:kfcf:lb}
	For every $\ell \geq 109$, there exists a $k$-fan-crossing-free graph \Gl with $n\in \Theta(\ell k)$ vertices such that $\crkfcf(\Gl) \in \Omega(n^2)$ and $\crn(\Gl) \in \O(k)$. Thus $\crr{\kfcf} \in \Omega(\frac{n^2}{k})$.
\end{theorem}
\begin{proof}
	Let \Gl be a framework graph, where the red con-graph is a \bundle{\widthRkfcfD}{}, blue con-graphs are $\bundle{\ell k}{3}$s, and gray con-graphs are $\bundle{\ell k}{2}$s, so $n\in \Theta(\ell k)$.
	A red-yellow standard drawing shows $\crn(\Gl) \in \O(k)$.
        Observe that in a blue con-graph, the ``middle'' edges of all the $\ell k$ paths are pairwise non-adjacent. We can thus obtain a fan-crossing-free drawing by considering a blue-blue standard drawing where all middle edges of the paths of one blue con-graph cross all middle edges of the paths of the other blue con-graph, cf.\ \Cref{fig:kfcf_blue_blue_cr}.

    Let \D be a $k$-fan-crossing-free drawing of \Gl.
    Each red or gray edge is adjacent to one of the six frame nodes.
    Therefore, the yellow edge $e$ can cross at most $2(k-1)$ of the red edges, so there exists an $s$-$t$-path $R$ in $S_\mathrm{red}$ not crossed by $e$.
    Let \Dred be the $R$-restriction of \D; it does not have any red-yellow crossings.
    Consider some edge $e$ in~\Dred.
    Let $\C=\{s,t\}$ be the connection with $e \in E(\C)$.
    For each of the four frame nodes $a \in N\setminus \{s,t\}$, edge $e$ can have at most $k-1$ crossings with gray edges incident to $a$.
    The crossings of $e$ with any gray edge cover at most $4(k-1) \cdot \frac{1}{\width{S_{\textrm{gray}}}\cdot \width{\C}} = \frac{4(k-1)}{\ell k \cdot\width{\C}}$ \otksRed (recall that the width of the red con-graph here in \Dred is $1$).
    As all con-graphs are bundles, there are $\sum_{\Calt \in \connections} \width{\Calt} \cdot h(\Calt)$ choices for $e$, so crossings involving a gray edge can cover at most $\sum_{\Calt \in \connections} \frac{4(k-1) \cdot \widthRed{\Calt} \cdot h(\Calt)}{\ell k \cdot \widthRed{\Calt}}$ \otksRed.
    Since $|\connections| = 9$, $\max_{\Calt \in \connections} \{h(\Calt)\} = 3$, and $\ell\geq109$, this simplifies to at most $9 \cdot \frac{4(k-1)\cdot 3}{\ell k} \leq \frac{108 (k-1)}{109k} \leq \frac{108}{109}$.
    Thus, blue-blue crossings have to cover at least $\frac{1}{109}$ \otksRed.
	As each such crossing covers at most $\O(\frac{1}{(\ell k)^2})$ \otksRed, we have $\crkfcf(\Gl) \in \Omega((\ell k)^2)=\Omega(n^2)$.
\end{proof}

In a non-$k$-fan-crossing-free drawing there has to be an edge that crosses $k\in \O(n)$ edges of a fan, so there are at least $k$ crossings. The $k$-fan-crossing-free drawing may require $\O(m^2)$ crossings, as by definition $k$-fan-crossing-free drawings are simple.
Therefore, \cref{thm:kfcf:lb} together with $\Rs[\kfcf]\in \O(\frac{n^2}{k})$ in the context of \Cref{thm:ub}(a), yields:
\begin{corollary}
	\label{thm:kfcf:crr}
	The $k$-fan-crossing-free crossing ratio $\crr{\kfcf}$ is in $\Theta(\frac{n^2}{k})$ and \ksharp.
\end{corollary}

Since NNIC-planarity is more restrictive than fan-crossing-free (which is $k$-FCF for $k=2$), but on the other hand NNIC-planar drawings are also necessarily simple, we also have:
\begin{corollary}
	\label{thm:nnic:crr}
	The NNIC crossing ratio $\crr{\NNIC}$ is in $\Theta(n^2)$.
\end{corollary}

A \emph{straight-line RAC} (right-angle crossings) drawing~\cite{RACPaper} is a drawing where edges are represented by straight lines and all crossings have to be at 90 degrees.
Consider any \Gl from \cref{thm:kfcf:lb}.
In a straight-line RAC drawing, there can be no edge crossing over two fan edges, so $\crbp[\slrac](\Gl) \geq \crkfcf[2](\Gl)$.
However, as witnessed by the blue-blue standard drawing (\Cref{fig:kfcf_blue_blue_cr}), each $\Gl$ is straight-line RAC drawable.
A straight-line RAC graph can have at most $4n-10$ edges \cite{RACPaper}, and thus no more than $\O(n^2)$ crossings in any simple drawing.
Therefore, our proof automatically also confirms \cite[Corollary 15]{BPS21}, without the need of an individual construction or proof:
\begin{corollary}
	\label{thm:slrac:crr}
	The straight-line RAC crossing ratio $\crr{\slrac}$ is in $\Theta(n^2)$.
\end{corollary}

\subsection{Fan-Planar and Variants\texorpdfstring{: $\boldmath\crr{\ac},\crr{\fc},\crr{\wfp},\crr{\sfp}$}{}}
\label{sec:fp}
A drawing is \emph{adjacency-crossing}~\cite{Adjacency/FanCrossing}, if no edge crosses two independent edges. 
In other words, whenever an edge $e$ crosses some edges $f_1,f_2$, the latter two have to have a common incident vertex.
A drawing is \emph{fan-crossing}~\cite{Adjacency/FanCrossing}, a slight further restriction of the former, if all edges $F=\{f_1,f_2,...\}$ that cross $e$ are incident to a common vertex; observe that the only difference to before is whether $e$ may cross all edges of a triangle.
A drawing is \emph{weakly fan-planar}, if all edges from $F$ furthermore cross $e$ ``from the same side'';
it is \emph{strongly fan-planar}, if also furthermore no endpoint of $e$ is enclosed by $e$ and the edges from $F$, see~\cite{kaufmannueckerdtjournal2022,cheong23} for details. 
In our notations, we use the shorthands ``$\ac$'', ``$\fc$'', ``$\wfp$'', and ``$\sfp$'' for the above concepts,  where $\crac(G)\leq\crfc(G)\leq\crwfp(G)\leq\crsfp(G)$ for any graph $G$ by definition. When considering straight-line drawings, all four concepts are identical.
\begin{theorem}\label{thm:fp:lb}
    For every $\ell \geq 1$, there exists a strongly fan-planar graph $G_\ell$ with $n\in \Theta(\ell)$ vertices such that $\crac(G_\ell)\in \Omega(n^2)$ and $\crn(G_\ell)\in \O(1)$. Thus $\crr{\ac},\crr{\fc},\crr{\wfp},\crr{\sfp}\in \Omega(n^2)$.
\end{theorem}
\begin{proof}
    Let $\Gl$ be a framework graph, in which the red and gray con-graphs are complete graphs on seven vertices ($K_7$) and the blue con-graphs are $\bundle{\ell}{2}$s, so $n\in \Theta(\ell)$. We need to describe how a $K_7$ con-graph is drawn in a standard drawing, see \Cref{fig:fp_yellow_red_cr}: Let $\C=\{s,t\}$ be the connection corresponding to a $K_7$ con-graph; as per definition, we draw the $K_7$ in a small neighborhood of $\C$'s curve such that $s,t$ are on the local outer face. Observe that $K_7$ can then still be drawn strongly fan-planarly. However, any adjacency-crossing drawing of a $K_7$ contains at least one $s$-$t$-path $P$ where every edge in $P$ is crossed by another edge of the $K_7$~\cite{BGDMPT2014}\footnote{We note that \cite[Lemma 7]{BGDMPT2014} is stated in terms of (weak) fan-planarity, but it is simple to see that the same argument even holds for fan- and adjacency-crossing drawings.}. We may call this property $\langle$P$\rangle$.
    A red-yellow standard drawing shows that $\crn(\Gl)\in \O(1)$. A blue-blue standard drawing, where all blue-blue crossings are adjacent to the same two frame nodes shows that the graph is strongly fan-planar, cf.\ \Cref{fig:fp_blue_blue_cr}.
    
    Consider an adjacency-crossing drawing $\D$ of $\Gl$. By property $\langle$P$\rangle$, we know that each red or gray $K_7$ con-graph gives rise to a pole-connecting path that cannot be crossed by any edge outside of its con-graph. Since the yellow edge is adjacent to both blue con-graphs, all \Kuras-subdivisions need to be covered by blue-blue crossings.
	Since each such crossing covers at most $\frac{1}{\ell^2}$ \otks, we have $\crac(G_\ell)\in \Omega(\ell^2)=\Omega(n^2)$.
\end{proof}

For $G\in \Gbp[\wfp]$, the upper bound $\crwfp[](G) \in \O(m^2)$, is attained by a simple drawing~\cite{klemz2023simplifying,klemz2023simplifyingA}.
In strongly fan-planar drawings, edges cannot cross multiple times (and there can be at most $\O(m^2)$ crossings between adjacent edges, if they exist at all), so the same asymptotic upper bound holds.
Therefore, \cref{thm:fp:lb} together with \cref{thm:ub}(a) yields:
\begin{corollary}
 \label{thm:fp:crrWS}  
The weakly and strongly fan-planar crossing ratios $\crr{\wfp},\crr{\sfp}$ are in $\Theta(n^2)$.
\end{corollary}

There are no results for the number of crossings a non-simple adjacency-crossing or fan-crossing drawing may have.
Since a simple drawing has at most $\O(m^2)$ crossings, \cref{thm:fp:lb} together with \cref{thm:ub}(a) yields:
\begin{corollary}
 \label{thm:fp:crrAF}  
The adjacency-crossing and fan-crossing crossing ratios $\crr[^{\text{simple}}(n)]{\ac},\crr[^{\text{simple}}(n)]{\fc}$ are in $\Theta(n^2)$ for simple drawings.
\end{corollary}

\subsection{\texorpdfstring{$\boldmath\text{}k$}{k}-Edge-Crossing\texorpdfstring{: $\boldmath\crr{\kecr}$}{}}
\label{sec:ecr}
A drawing is \emph{$k$-edge-crossing} (shorthand ``\kecr'') if at most $k$ edges have crossings~\cite{I08,GSM10}.

\begin{theorem}
	\label{thm:ecr:lb}
	For every $\ell \geq 1$, there exists a $k$-edge-crossing graph \Gl with $n \in \Theta(\ell k)$ vertices such that $\crkecr(\Gl) \in \Omega(k^2)$ and $\crn(\Gl) \in \O(k)$. Thus $\crr{\kecr} \in \Omega(k)$.
\end{theorem}
\begin{proof}
	Let \Gl be a framework graph, where the red con-graph is a $\bundle{\kecrParam}{}$, the gray con-graphs are $\bundle{\ell k}{}$s, and blue con-graphs are $\bundle{\frac{k}{2}}{}$s, so $n\in \Theta(\ell k)$.
    A red-yellow standard drawing shows that $\crkecr(\Gl) \in \O(k)$.
    A blue-blue standard drawing shows that \Gl is $k$-edge-crossing, cf.\ \cref{fig:fp_blue_blue_cr}.
   
    Consider a drawing \D of \Gl with at most $k$ crossed edges.
    A con-graph is \emph{cut} if all paths between its poles are crossed.
    There are at least two cut con-graphs in \D, as otherwise, it would contain non-covered \Kuras-subdivisions.
    Each cut con-graph contains at least its width many crossed edges.
    Only yellow and blue con-graphs have widths less than $k$, so only they can be cut without exceeding $k$ crossed edges.
    Therefore, at least $\frac{k}{2} + 1$ crossed edges are blue or yellow. Consequently, at most $\frac{k}{2}-1$ crossed edges are red or gray.
    Any crossing of these red or gray edges with the unique yellow edge can cover at most $\frac{1}{k}$ \otks, so they together cover at most $(\frac{k}{2}-1)\frac{1}{k}\leq\frac12$ \otks.
    The remaining at least $\frac12$ \otks have to be covered by other crossings, but they  cover at most $\O(\frac{1}{k^2})$ \otks each. This yields $\crkecr(\Gl) \in \Omega(k^2)$. 
\end{proof}

In a non-\kecr drawing there are at least $\frac{k}{2}$ crossings as there are at least $k+1$ crossing edges.
In a \kecr drawing there are less than $\frac{k^2}{2}\in \O(mk)$ crossings as each of the at most $k$ crossing edges can cross each other crossing edge, leading to $\Rs[\kecr]\in \O(\frac{k^2}{k})=\O(k)$ in the context of \Cref{thm:ub}(b); together with \cref{thm:ecr:lb} we yield:
\begin{corollary}
	\label{thm:ecr:crr}
	The $k$-edge crossing crossing ratio $\crr{\kecr}$ is in $\Theta(k)$ and \ksharp.
\end{corollary}

\subsection{\texorpdfstring{$\boldmath\text{}k$}{k}-Gap-Planar \texorpdfstring{: $\boldmath\crr{\kgappl}$}{}}
\label{sec:kgappl}
Let $E(x)$ denote the two edges involved in a crossing $x$.
A drawing is \emph{$k$-gap-planar} (shorthand ``\kgappl'') if there exists a mapping from each crossing $x$ to one of $E(x)$, such that any edge $e\in E(G)$ is mapped to at most $k$ times overall~\cite{bae2018gap,bae2018gapA}. The intuition (and reason for the name) is that we may want to visualize crossings by drawing a ``gap'' on one of the two involved edge curves, but each edge curve may not attain too many gaps. 
\begin{figure}
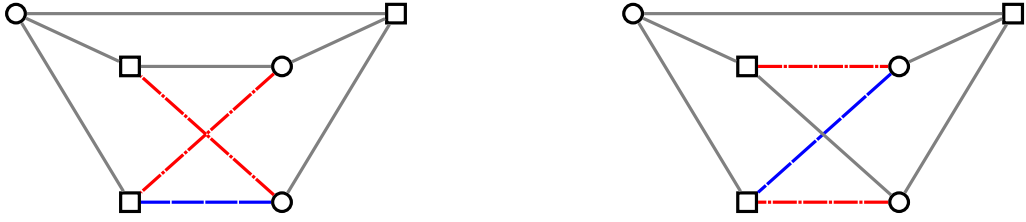

	\centering
	\begin{subfigure}[b]{0.42\textwidth}
        \frameDrawing{blueE}{gray}{redE}{redE}{}{}
	\end{subfigure}
 \hspace{\fill}
	\begin{subfigure}[b]{0.42\textwidth}
        \frameDrawing{redE}{redE}{blueE}{grayE}{}{}
	\end{subfigure}
	\caption{Drawings of the differently colored frame used in \cref{thm:kgap:lb,thm:kapex:lb,thm:skewk:lb} for the $k$-gap-planar, $k$-apex-planar, and skewness-$k$ crossing ratios, respectively.} 
	\label{fig:alt_frame}
\end{figure}

\begin{theorem}\label{thm:kgap:lb}
	For every $\ell \geq 5$, there exists a $k$-gap-planar graph \Gl with $n \in \Theta(\ell k)$ vertices such that $\crbp[\kgappl](\Gl) \in \Omega(nk)$ and $\crn(\Gl) \in \O(k^2)$. Thus $\crr{\kgappl} \in \Omega(\frac{n}{k})$.
\end{theorem}
\begin{proof}
    Consider a coloring of the frame $F$ where a blue connection is adjacent to two independent red connections, and all other connections are gray, see \Cref{fig:alt_frame}.
    Let \Gl be a framework graph, where the blue and red con-graphs are $\bundle{5k}{}$s, and the gray con-graphs are $\bundle{\ell k}{5}$s, so $n\in \Theta(\ell k)$.
	A red-red standard drawing shows $\crn(\Gl) \in \O(k^2)$.
	A blue-gray standard drawing shows that \Gl is $k$-gap-planar, cf.\ \cref{fig:kgap_blue_gray_cr}.

    Let \D be a $k$-gap-planar drawing of \Gl.
    A red-red crossing covers at most $\frac{1}{(5k)^2}=\frac{1}{25k^2}$ \otks, but the $2 \cdot 2 \cdot 5k$ red edges can have at most $20k \cdot k=20k^2$ crossings between each other.
    As the blue connection is adjacent to the red connections, at least $1-\frac{20 k^2}{25 k^2}=\frac{1}{5}$ \otks have to be resolved by crossings involving a gray edge. Since such a crossing can cover at most $\O(\frac{1}{\ell k\cdot k})$ \otks, we have 
    $\crbp[\kgappl](\Gl) \in \Omega(\ell k^2)=\Omega(nk)$.
\end{proof}

A non-$k$-gap-planar drawing has over $k^2$ crossings, since there is an
edge that has at least 
$k+1$ crossings mapped to it, and all the corresponding crossing edges have at least $k$ crossings mapped to them.
In a $k$-gap-planar drawing, there can be at most $mk$ crossings as at most $k$ crossings can be mapped to each of the $m$ edges, leading to $\Rs[\kgappl]\in \O(\frac{mk}{k^2})=\O(\frac{n}{k})$ in the context of \Cref{thm:ub}(b); together with \cref{thm:kgap:lb} we yield:
\begin{corollary}
	\label{thm:kgappl:crr}
	The $k$-gap-planar crossing ratio $\crr{\kgappl}$ is in $\Theta(\frac{n}{k})$.%
\end{corollary}

\subsection{\texorpdfstring{$\boldmath\text{}k$}{k}-Apex-Planar and Skewness-\texorpdfstring{$\boldmath\text{}k$}{k}\texorpdfstring{: $\boldmath\crr{\kapex}$ and $\boldmath\crr{\skewk}$}{}}
\label{sec:kapex}
A drawing is \emph{$k$-apex-planar} (\emph{skewness-$k$})---with shorthand ``$k$-apex'' (``skew-$k$'')---if there are $k$ vertices, called \emph{apex vertices} (edges, called \emph{skewness edges}, respectively), whose removal yields a planar subdrawing.
Let \bp be either of the two concepts.
To prove crossing ratio \emph{upper} bounds, \cite{BPS21} furthermore needs to assume that a crossing minimal \bp-drawing is simple (which is, unfortunately, not true in general).
We do not want to impose the simplicity requirement, and for our bounds it suffices:
\begin{lemma}\label{lem:simple}
    Let $\bp\in\{\kapex, \skewk\}$.
    All \bp-graphs allow simple \bp-drawings. Thus, the (non-simple) \bp-crossing number is always upper bounded by a well-defined simple \bp-crossing number.
\end{lemma}
\begin{proof}
    Let $R$ be the edges whose removal leaves a planar subgraph $G'$ (we retain apex vertices, all of which now have degree~0).
    Draw $G'$ planar with straight lines. 
    Now add all edges $R$ as straight lines. The resulting straight-line drawing (possibly after $\varepsilon$-perturbations to avoid collinearities) is a simple $k$-apex-planar (or skewness-$k$) drawing by definition.
\end{proof}

\begin{theorem}
	\label{thm:kapex:lb}
	For every $\ell \geq 1$, there exists a $k$-apex-planar graph \Gl with $n \in \Theta(\ell k)$ vertices such that $\crbp[\kapex](\Gl) \in \Omega(n^2)$ and $\crn(\Gl) \in \O(k)$. Thus $\crr{\kapex} \in \Omega(\frac{n^2}{k})$.
\end{theorem}
\begin{proof}
    Consider a coloring of the frame $F$ where a blue connection is adjacent to two independent red connections, and all other connections are gray, see \Cref{fig:alt_frame}.
    Let \Gl be a framework graph, where
    the red con-graphs are single edges, the gray con-graphs are $\bundle{\ell k}{2}$s, but the blue con-graph is more involved:
    Start with a \bundle{k}{} $B$ and let $A$ be its $k$ degree-$2$ vertices;
    replace every edge of $B$ with a \bundle{\ell}{};
    finally, for each $v\in A$, add four new vertices that form a $K_5$ together with $v$.
    This yields $n\in \Theta(\ell k)$. 
    For standard drawings, we draw the blue con-graph without the $K_5$s planarly and then, for each $v\in A$, we draw its $K_5$ in a small neighborhood of $v$ with one crossing that is adjacent to $v$.
    In such a drawing, there are $k$ crossings and we have to remove $k$ vertices (for example~$A$) to obtain a planar subdrawing.
    A red-red standard drawing shows that $\crn(\Gl) \leq k+1$.
    A blue-gray standard drawing with all crossings incident to a vertex of $A$ shows that $\Gl$ is $k$-apex-planar, cf.\ \Cref{fig:kapex_blue_gray_cr}.
   
    Consider a $k$-apex-planar drawing \D of \Gl.
    Each of the $k$ disjoint $K_5$ subgraphs in $S_\mathrm{blue}$ needs to be covered by a crossing.
    Thus, all $k$ apex nodes are in $S_\mathrm{blue}$ and so all \otks are covered by blue-gray crossings.
    As each such crossing can cover at most $\frac{1}{(\ell k)^2}$ \otks, we have
    $\crbp[\kapex](\Gl) \in \Omega((\ell k)^2)=\Omega(n^2)$.
\end{proof}

Any non-\kapex-planar drawing has at least $k+1$ crossings; by \Cref{lem:simple} we have $\Rs[\kapex]\in \O(\frac{m^2}{k})=\O(\frac{n^2}{k})$ in the context of \Cref{thm:ub}(a).
With \cref{thm:kapex:lb} we yield\rlap{:}
\begin{corollary}\label{thm:kapex:crr}
    The $k$-apex-planar crossing ratio $\crr{\kapex}$ is in $\Theta(\frac{n^2}{k})$ and \ksharp.
\end{corollary}

\begin{theorem}
	\label{thm:skewk:lb}
	For every $\ell > k$, there exists a skewness-$k$ graph \Gl with $n \in \Theta(\ell k)$ vertices such that $\crbp[\skewk](\Gl) \in \Omega(nk)$ and $\crn(\Gl) \in \O(k)$. Thus $\crr{\skewk} \in \Omega(n)$.
\end{theorem}
\begin{proof}
    Consider the same coloring of the frame $F$ as in \Cref{thm:kapex:lb}, see \Cref{fig:alt_frame}.
    Let \Gl be a framework graph, where the red con-graphs are single edges, the gray con-graphs are $\bundle{\ell k}{2}$s, but the blue con-graph is more involved:
    Start with a \bundle{k}{2} $B$ and let $A$ be its $k$ degree-$2$ vertices and $s$ and $t$ its poles;
    for each $v\in A$, add three new vertices $W_v$ that form a $K_5$ together with $v$ and $s$, but remove the edge $\{s,v\}$.
    This yields $n \in \Theta(\ell k)$.
    Let $Q$ be the $\ell k$ edge-disjoint $s$-$t$-paths of gray edges in \Gl. 
    We observe $k$ edge-disjoint $K_5$ subdivisions, one for each $v$, induced by the vertices $\{s,v\}\cup W_v$ together with the edge $\{v,t\}$ and a path of $Q$ to establish the $v$-$s$-Kuratowski path.
    In standard drawings, we draw the blue con-graph with $k$ crossings, one for each $v\in A$, such that $\{v,t\}$ is crossed by an edge between vertices of $W_v$.
    We can remove the $k$ blue edges incident to $t$ to obtain a planar blue subdrawing.
    A red-red standard drawing shows that $\crn(\Gl) \leq k+1$.
    A blue-gray standard drawing, where each crossing involves a blue edge incident to $t$, shows that $\Gl$ is skewness-$k$, cf.\ \Cref{fig:skewk_blue_gray_cr}.
   
    Consider a skewness-$k$ drawing \D of \Gl.
    For each $v\in A$, consider the $\ell k$ $K_5$ subdivisions that differ only by the chosen path of $Q$. Since there cannot be $\ell k$ (gray) skewness edges, we have a blue skewness edge per $v$, which is not in any considered $K_5$ subdivision for any other vertex in $A$. 
    Thus, all $k$ skewness edges are blue and all \otks have to be covered by  blue-gray crossings.
    As each such crossing can cover at most $\frac{1}{\ell k^2}$ \otks, we have
    $\crbp[\kapex](\Gl) \in \Omega(\ell k^2)=\Omega(nk)$.
\end{proof}

A non-skewness-$k$ drawing has at least $k+1$ crossings.
A simple skewness-$k$ drawing has at most $mk$ crossings, attained when each of the $k$ edges we want to remove crosses all other edges.
Via \Cref{lem:simple}, this leads to $\Rs[\skewk]\in \O(\frac{mk}{k}) = \O(n)$ for \cref{thm:ub}(b). Thus
\begin{corollary}
	\label{thm:skewk:crr}
	The skewness-$k$ crossing ratio $\crr{\skewk}$ is in $\Theta(n)$ and \ksharp.
\end{corollary}

\subsection{Rectilinear Crossing Ratios\texorpdfstring{: $\boldmath\text{}\bcrrbp$}{}}
\label{sec:rectilinear}
The rectilinear crossing number $\bcr(G)\geq \crn(G)$ is the minimum number of crossings in a straight-line drawing of $G$~\cite{rectilinearCrossingNumber}. Analogously, we can define $\bcrbp(G)\geq \crbp(G)$ as the minimum number of crossings over all straight-line drawings of $G$ subject to the beyond planarity concept $\bp$. We naturally yield the rectilinear crossing ratio $\bcrrbp$, defined analogously to $\crrbp$ but w.r.t.\ the rectilinear crossing numbers.

All previous crossing number \emph{lower} bounds still hold in the straight-line setting. 
Further, we establish the crossing number \emph{upper} bounds via straight-line drawings in nearly all above theorems---only for the fan-planar variants, its $K_7$s require non-straight edges. Thus:
\begin{corollary}
    The rectilinear crossing ratio bounds for $k$-planar, $k$-vertex-planar, IC-planar, NIC-planar, NNIC-planar, $k$-fan-crossing-free, $k$-gap-planar, $k$-edge-crossing, skewness-$k$, and $k$-apex-planar graphs coincide with those of the above crossing ratios.
\end{corollary}

\section{Conclusion}%

We presented a framework to yield (short) proofs for the crossing ratios of many different beyond-planarity concepts, summarized in \Cref{tab:res}. We also included straight-line and non-simple variants of these concepts. Most importantly, and in contrast to most previous approaches, this allows us to attain asymptotically tight bounds.
The key idea is to consider graphs arising from a single $K_{3,3}$ by ``thickening'' edges via subgraphs, and than count how many of the arising 
$K_{3,3}$-subdivisions can be resolved by certain crossings, yielding a lower bound on the total number of crossings in beyond-planarity-restricted drawings.

There are only three beyond-planarity concepts which we did not discuss although their crossing ratios (on simple drawings) have previously been studied:
For \emph{planarly connected} graphs, the tight bound $\Theta(n^2)$ is already known~\cite{BPS21}.
For \emph{$k$-quasi-planar} graphs, our \cref{thm:ub}(a) easily yields the improved upper bound $\crr{k\text{-qp}} \in \O(n^2/ k^2)$ instead of $\O(f(k) \cdot n^2 (\log n)^2)$~\cite{CKMV19} for some growing function $f$.
A \emph{$(k,l)$-grid-free} drawing, $l \leq k$, can have at most $O(k^{1/l} \cdot m^{2-1/l})$ crossings~\cite{KST54}.
Thus, for $(k,l)$-grid-free graphs, our \cref{thm:ub}(a) yields the upper bound $\crr{(k,l)\text{-grid-free}} \in \O(n^{2-1/l}/(lk^{1-1/l}))\subset \O(n^2/l)$ instead of $\O(g(k,l) \cdot n^2)$~\cite{BPS21}, where $g$ is exponential in $l$.
The idea of counting crossings that resolve certain percentages of Kuratowski subdivisions seems to allow us to also improve, at least slightly, the known crossing ratio lower bounds for both.
However, this cannot easily be done purely within our framework and the lower bounds remain linear in $n$ and therefore far from matching the upper bounds.

Interestingly, for every but one considered beyond-planarity concept \bp, we were able to show that the ratio is \emph{\ksharpness}, i.e., we can observe the worst-case crossing ratio already for graphs that would attain the optimal crossing number if we were to increase $k$ by just~1.
Only for the $k$-gap-planar crossing ratio, our construction does not yield \ksharpness{};
it is unclear whether this is a limitation of our framework or if the $k$-gap-planar crossing ratio is indeed not \ksharpness.

\bibliography{bib2doi}

\clearpage
\appendix

\section{APPENDIX: Fan-crossing-free vs.\ NNIC-planar}

We give the full proof (omitted above) to show that FCF and NNIC are different concepts.

\begin{figure}[b!]
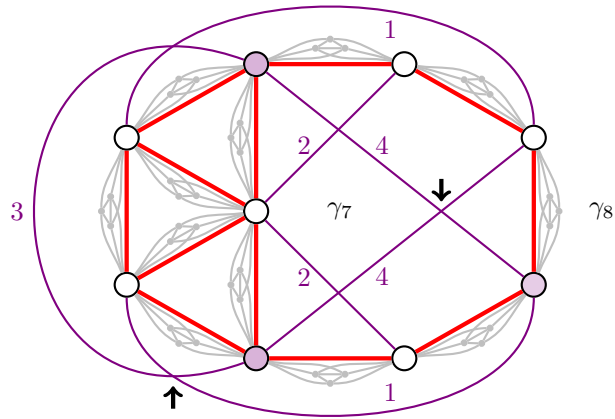

\centering
        \FCFvsNNICdrawing{true}
    \caption{\nolinenumbers{}Labeled version of \Cref{fig:nnic_vs_fcf}: A fan-crossing free but not NNIC-planar graph $G=(V_L\cup V_S, E_W\cup E_G\cup E_L)$. Large vertices $V_L$ are drawn as large circles, small vertices $V_L$ as small gray dots. Wall edges $E_W$ are drawn red; guards $E_G$ gray; loners $E_L$ violet. Once it is established that walls can never be crossed, the labels at the loners gives the order used to argue that the fan-crossing-free embedding is unique (up to mirroring and irrelevant details for the guards). The embedding is not NNIC-planar: arrows point to two crossings that share three vertices (shaded in light violet).}
    \label{fig:nnic_vs_fcf_extended}
\end{figure}

\THMnnicfcf*

\begin{proof}
We will establish that the graph $G$ given in \Cref{fig:nnic_vs_fcf_extended} is fan-crossing-free but not NNIC-planar. To understand the graph better, we partition the vertices of $G$ into \emph{large} vertices $V_L$ and \emph{small} vertices $V_S$, as depicted. We partition the edges into \emph{walls} $E_W$, \emph{guards} $E_G$, and \emph{loners} $E_L$.
Guards are precisely those edges that are adjacent to a small vertex. Let $e$ be a wall, then $K_e$ is the unique $K_5$ consisting of $e$ and nine guards. Observe that $V_S$ are the only vertices of degree $4$, while all $V_L$ vertices have degree at least $9$.
In the following, let $\mathcal{D}$ be a fan-crossing-free drawing of $G$ (for example, but not necessarily precisely the one depicted in \Cref{fig:nnic_vs_fcf_extended}). Let $\mathcal{D}_e$ be the (fan-crossing-free) subdrawing of $K_e$, for any wall $e$.
It is easy to see that any fan-crossing-free drawing of a $K_5$ has exactly one crossing, and each of its $8$ cells (regions bounded by edges, vertices and crossing points) has degree $3$ ($3$ parts of different edge curves). Furthermore, any edge can cross at most one edge of a $K_e$ without a fan-crossing, and so its end points have to reside in two distinct cells of $\mathcal{D}_e$ (property $\boldmath\langle\textbf{A}\rangle$).

Let $e,f\in E_W$ be two walls, and assume that there is a crossing between some edge $e'\in E(K_e)$ and edge $f'\in E(K_f)$. Consider the cells of $\mathcal{D}_f$.
If there would be a cell $\gamma$ of $\mathcal{D}_f$ that contains only one of the vertices $V(K_e)$ (say $v$), the $4$ $v$-incident edges of $E(K_e)$ would have to leave $\gamma$. Since $\gamma$ has degree 3, by pigeonhole principle there would be an edge bordering $\gamma$ crossed by a fan (two $v$-incident edges); a contradiction.
Thus, let $\gamma_1$ be the $\mathcal{D}_f$-cell containing two vertices $v_1,v_2\in V(K_e)$, and $\gamma_2$ the (different) $\mathcal{D}_f$-cell containing the other three $U_3$. Each of the three $v_1$-incident edges connecting to $U_3$ has to cross a different edge bordering $\gamma_1$. But then, by property $\langle\text{A}\rangle$, neither of them may cross another edge of $K_f$ without introducing a fan-crossing. So, the three edges cannot all end at a common cell $\gamma_2$.

Thus, we know that there is no crossing between any edges of the different $K_e$, for $e\in E_W$. Furthermore, observe that in $G[W_E]\subset G$, the subgraph containing only the walls, there is no edge whose end points form a separation pair (i.e., removing any pair of adjacent vertices in $G[W_E]$ leaves a connected subgraph). Thus in $\mathcal{D}$, for each $e\in W_E$, all of $G\setminus K_e$ is drawn in a common cell of $\mathcal{D}_e$ (which we may call the outer cell; property $\boldmath\langle\textbf{B}\rangle$), and $e$ is also not crossed by any edge of $K_e$.

Overall, the subdrawing of \D of subgraph $G[W_E]$ has to be a planar and its embedding is unique, up to mirroring. Let $\gamma_8$ and $\gamma_7$ be the unique largest and second-largest face of this embedding.
Since loners $E_L$ always connect large vertices, they cannot cross any wall or guard edges by properties $\langle\text{A}\rangle$ and $\langle\text{B}\rangle$.

Consider the loners $E_L$ in the order specified in \Cref{fig:nnic_vs_fcf_extended}: the four loners labeled ``1'' or ``2'' have a unique embedding that yields no crossing with wall edges. Now, the loner labeled ``3'' has to be routed through $\gamma_8$ (crossing over the two non-adjacent loners labeled ``1''), as routing through $\gamma_7$ would cross the fan of the two loners labeled ``2''. Now, the loners labeled ``4'' cannot be routed through $\gamma_8$ (each would form a fan together with ``3'' that is crossed by a loner labeled ``1''), and have to be routed through $\gamma_7$ to establish a fan-crossing-free drawing. In particular, the constructed fan-crossing-free subembedding of $G[E_W\cup E_L]$ is unique up to mirroring.
But it contains even two pairs of crossings, each of which shares three vertices---a contradition to NNIC-planarity.
\end{proof}

\end{document}

%% file: graphs.tex
\newcommand{\frameDrawing}[7][grayE]{
    \begin{tikzpicture}[very thick]
        \path[use as bounding box] (2.7,1.9) rectangle (-2.7,-1.4);
        \node[draw,       inner sep=2pt,minimum size=7pt] (tl) at ( -1.0,  0.8) {};
        \node[draw,       inner sep=2pt,minimum size=7pt] (bl) at ( -1.0, -1.0) {};
        \ifstrequal{#7}{}{
            \node[draw,circle,inner sep=2pt,minimum size=7pt] (tr) at (  1.0,  0.8) {};
        }{
            \node[draw,circle,inner sep=2pt,minimum size=7pt] (tr) at #7 {};
        }
        \node[draw,circle,inner sep=2pt,minimum size=7pt] (br) at (  1.0, -1.0) {};
        \node[draw,circle,inner sep=2pt,minimum size=7pt] (l)  at (-2.5,  1.5) {};
        \node[draw,       inner sep=2pt,minimum size=7pt] (r)  at ( 2.5,  1.5) {};
        
        \path[grayE,draw] (l) to (r);
        
        \draw[grayE] (l) -- (tl);
        \draw[grayE] (l) -- (bl);
        \draw[grayE] (r) -- (br);
        \draw[#1]   (r) -- (tr);
        
        \draw[#2]  (br) -- (bl);%
        \draw[#3]  (tl) -- (tr);%
        \draw[#4]  (tr) -- (bl);%
        \draw[#5]  (br) -- (tl);%
    \end{tikzpicture}
}

\newcommand{\FCFvsNNICdrawing}[1]{
    \begin{tikzpicture}[thick,scale=1.3]
        \node[draw,circle,fill=violet!20,inner sep=1pt, minimum size=9pt] (i) at (  2.81, -0.75 )  {}; %
        \node[draw,circle,               inner sep=1pt, minimum size=9pt] (j) at (  2.81,  0.75 )  {}; %
        \node[draw,circle,               inner sep=1pt, minimum size=9pt] (a) at (  1.5 , -1.5 )  {}; %
        \node[draw,circle,               inner sep=1pt, minimum size=9pt] (e) at (  1.5 ,  1.5 )  {}; %
        \node[draw,circle,fill=violet!30,inner sep=1pt, minimum size=9pt] (b) at ( -0.0 , -1.5 )  {}; %
        \node[draw,circle,fill=violet!30,inner sep=1pt, minimum size=9pt] (d) at ( -0.0 ,  1.5 )  {}; %
        \node[draw,circle,               inner sep=1pt, minimum size=9pt] (h) at ( -0.0 ,  0.0 )  {}; %
        \node[draw,circle,               inner sep=1pt, minimum size=9pt] (c) at ( -1.31, -0.75)  {}; %
        \node[draw,circle,               inner sep=1pt, minimum size=9pt] (f) at ( -1.31,  0.75)  {}; %
        
        \ifstrequal{#1}{true}{
            \node[minimum size=5pt,fill=white,circle,inner sep=0.2pt] at (0.85,0.0) {$\gamma_7$};
            \node[minimum size=5pt] at (3.5,0.0) {$\gamma_8$};
        }{}

        \foreach \u/\v in {d/f, e/d, j/e, i/j, a/i, b/a, c/b, h/c, h/b, d/h, f/h, f/c}{
            \miniKfive{(\v)}{(\u)}
        }

        \ifstrequal{#1}{true}{
            \path[draw, name path=IC, violet, bend right=90, looseness=1.0] (c) to node[pos=0.6,above]{1} (i);
            \path[draw, name path=JF, violet, bend right=90, looseness=1.0] (j) to node[pos=0.4,below]{1} (f);
            \path[draw, name path=BD, violet,bend left=110,looseness=2.5] (b) to node[left]{3} (d);
            \draw[violet] (a) -- (h) node[pos=0.7,below]{2};
            \draw[violet] (e) -- (h) node[pos=0.7,above]{2};
            \path[draw, name path=BJ, violet] (b) -- (j) node[pos=0.45,below] {4};
            \path[draw, name path=DI, violet] (d) -- (i) node[pos=0.45,above] {4};
        }{
            \path[draw, name path=IC, violet, bend right=90, looseness=1.0] (c) to (i);
            \path[draw, name path=JF, violet, bend right=90, looseness=1.0] (j) to (f);
            \path[draw, name path=BD, violet,bend left=110,looseness=2.5] (b) to (d);
            \path[draw, name path=BJ, violet] (b) -- (j);
            \path[draw, name path=DI, violet] (d) -- (i);
            \draw[violet] (a) -- (h);
            \draw[violet] (e) -- (h);
        }
        
        \path[name intersections={of=DI and BJ, by=BJDI}];
        \node[minimum size=5pt] (BJDInode) at (BJDI) {};
        \draw[<-,ultra thick] (BJDInode) -- +(0,0.33);
        \path[name intersections={of=IC and BD, by=ICBD}];
        \node[minimum size=5pt] (ICBDnode) at (ICBD) {};
        \draw[<-,ultra thick] (ICBDnode) -- +(0,-0.33);
    \end{tikzpicture}
}

\newcommand{\miniKfive}[2]{
    \node[draw,fill, circle, inner sep=0pt, minimum size=1.8pt,gray!50] (u) at ($#2!.6!#1!.27!-90:#1$) {};
    \node[draw,fill, circle, inner sep=0pt, minimum size=1.8pt,gray!50] (v) at ($#1!.5!#2!.34!90:#2$) {};
    \node[draw,fill, circle, inner sep=0pt, minimum size=1.8pt,gray!50] (w) at ($#1!.6!#2!.27!90:#2$) {};
    \draw[gray!50,relative,in=180,out= 3] #1 to (u);
    \draw[gray!50,relative,in=180,out=-3] #2 to (w);
    \draw[gray!50] (u) -- (v) -- (w) -- (u);
    \draw[gray!50,relative,in=-160,out=-3] #2 to (v);
    \draw[gray!50,relative,in= 160,out=-3] #2 to (u);
    \draw[gray!50,relative,in=-160,out= 3] #1 to (w);
    \draw[gray!50,relative,in= 160,out= 3] #1 to (v);
    \draw[red, ultra thick,relative,in=-180,out=-0] #1 to #2;
}

\newcommand{\crossbundle}[6]{
	\def\scl{0.2}
	\begin{tikzpicture}
		\pgfmathparse{#1+#4}
		\xdef\maxW{\pgfmathresult}
		\pgfmathparse{#2+#3}
		\xdef\maxH{\pgfmathresult}
		\node[draw, circle, inner sep=1pt,minimum size=3pt] (u) at ($\scl*(\maxW*0.5,   \maxH*1+1.5)$) {\footnotesize $s_1$};
		\node[draw, circle, inner sep=1pt,minimum size=3pt] (v) at ($\scl*(\maxW*0.5,   -1.5)$)        {\footnotesize $t_1$};
		\node[draw, circle, inner sep=1pt,minimum size=3pt] (w) at ($\scl*(-1.5,        \maxH*0.5)$)   {\footnotesize $s_2$};
		\node[draw, circle, inner sep=1pt,minimum size=3pt] (z) at ($\scl*(\maxW*1+1.5, \maxH*0.5)$)   {\footnotesize $t_2$};

		\pgfmathparse{\maxW/#4}
		\xdef\xRatio{\pgfmathresult}
		\pgfmathparse{1#6}
		\xdef\xSum{\pgfmathresult}
		\foreach \xCnt [evaluate=\xCnt as \x using \xCnt*1,parse=true] in {1,2,...,\maxW-1}{
			\pgfmathparse{\xSum >= \xRatio}
			\ifthenelse{\equal{\pgfmathresult}{1}}{
				\xdef\xDraw{1}
				\pgfmathparse{\xSum - \xRatio+1}
				\xdef\xSum{\pgfmathresult}
			}{
				\xdef\xDraw{0}
				\pgfmathparse{\xSum + 1}
				\xdef\xSum{\pgfmathresult}
			}
			\pgfmathparse{\maxH/#2}
			\xdef\yRatio{\pgfmathresult}
			\pgfmathparse{1#5}
			\xdef\ySum{\pgfmathresult}
			\foreach \yCnt [evaluate=\yCnt as \y using \yCnt*1,parse=true] in {1,2,...,\maxH-1}{
				\pgfmathparse{\ySum >= \yRatio}
				\ifthenelse{\equal{\pgfmathresult}{1}}{
					\pgfmathparse{\ySum - \yRatio}
					\xdef\ySum{\pgfmathresult}
				}{
					\ifthenelse{\equal{\xDraw}{1}}{
						\node[fill=black, circle, scale=1.0*\scl] (wz\xCnt\yCnt) at ($\scl*(\x,\y)$) {};
					}{}
					\ifthenelse{\equal{\xCnt}{1}}{
						\path[rounded corners=\scl*10,draw] (w) -- ($\scl*(0.25,    \y)$) -- ($\scl*(\x,\y)$);
					}{
						\path[rounded corners=\scl*10,draw] ($\scl*(\x-1, \y)$) -- ($\scl*(\x,\y)$);
					}
					\pgfmathparse{\xCnt == \maxW-1}
					\ifthenelse{\equal{\pgfmathresult}{1}}{
						\path[rounded corners=\scl*10,draw] (z) -- ($\scl*(\maxW-0.25,\y)$) -- ($\scl*(\x,\y)$);
					}{}
				}
				\pgfmathparse{\ySum + 1}
				\xdef\ySum{\pgfmathresult}
			}
		}
		\pgfmathparse{\maxH/#2}
		\xdef\xRatio{\pgfmathresult}
		\pgfmathparse{1#5}
		\xdef\xSum{\pgfmathresult}
		\foreach \xCnt [evaluate=\xCnt as \x using \xCnt*1,parse=true] in {1,2,...,\maxH-1}{
			\pgfmathparse{\xSum >= \xRatio}
			\ifthenelse{\equal{\pgfmathresult}{1}}{
				\xdef\xDraw{1}
				\pgfmathparse{\xSum - \xRatio+1}
				\xdef\xSum{\pgfmathresult}
			}{
				\xdef\xDraw{0}
				\pgfmathparse{\xSum + 1}
				\xdef\xSum{\pgfmathresult}
			}
			\pgfmathparse{\maxW/#4}
			\xdef\yRatio{\pgfmathresult}
			\pgfmathparse{1#6}
			\xdef\ySum{\pgfmathresult}
			\foreach \yCnt [evaluate=\yCnt as \y using \yCnt*1,parse=true] in {1,2,...,\maxW-1}{
				\pgfmathparse{\ySum >= \yRatio}
				\ifthenelse{\equal{\pgfmathresult}{1}}{
					\pgfmathparse{\ySum - \yRatio}
					\xdef\ySum{\pgfmathresult}
				}{
					\ifthenelse{\equal{\xDraw}{1}}{
						\node[fill=black, circle, scale=1.0*\scl] (uv\xCnt\yCnt) at ($\scl*(\y,\x)$) {};
					}{}
					\ifthenelse{\equal{\xCnt}{1}}{
						\path[rounded corners=\scl*10,draw] (v) -- ($\scl*(\y,    0.25)$) -- ($\scl*(\y,\x)$);
					}{
						\path[rounded corners=\scl*10,draw] ($\scl*(\y, \x-1)$) -- ($\scl*(\y,\x)$);
					}
					\pgfmathparse{\xCnt == \maxH-1}
					\ifthenelse{\equal{\pgfmathresult}{1}}{
						\path[rounded corners=\scl*10,draw] (u) -- ($\scl*(\y,\maxH*1-0.25)$) -- ($\scl*(\y,\x)$);
					}{}
				}
				\pgfmathparse{\ySum + 1}
				\xdef\ySum{\pgfmathresult}
			}
		}
	\end{tikzpicture}
}